%% file: privFS.tex
\documentclass{sig-alternate}
\usepackage{amsmath, alltt, amssymb, xspace, times, epsfig, comment, url, graphicx}
\usepackage{amssymb, comment, algorithm, algpseudocode, mathtools}
\usepackage{verbatim}
\usepackage{multirow}
\usepackage{color}

\newtheorem{theorem}{Theorem}
\newtheorem{lemma}{Lemma}

\newtheorem{definition}{Definition}
\newtheorem{example}{Example}[section]

\newcommand\todo[1]{\{\textbf{Todo:} \textit{#1}\}}

\newcommand\note[1]{\{\textbf{Note:} \textit{#1}\}}

\newcommand{\mypara}[1]{\vspace*{0.06in}\noindent\textbf{#1} \xspace}
\newcommand{\myspara}[1]{\vspace*{0.06in}\noindent\emph{#1} \xspace}

\newcommand{\difp}{differential privacy\xspace}
\newcommand{\chitest}{\ensuremath{\chi^2} correlation test\xspace}

\newcommand{\Lap}[1]{\ensuremath{\mathsf{Lap}\left(#1\right)}\xspace}

\newcommand{\cor}[1]{\ensuremath{\mathsf{Cor}\left(#1\right)}\xspace}
\newcommand{\HG}[1]{\ensuremath{\mathsf{H}\left(#1\right)}\xspace}
\newcommand{\HS}{\ensuremath{\mathsf{selectHist}}\xspace}
\newcommand{\HP}{\ensuremath{\mathsf{perturbHist}}\xspace}
\newcommand{\FS}{\ensuremath{\mathsf{selectFeature}}\xspace}

\renewcommand{\Pr}[1]{\ensuremath{\mathsf{Pr}\left[#1\right]}\xspace}
\newcommand{\myexp}[1]{\ensuremath{e^{#1}}\xspace}

\renewcommand{\AA}{\ensuremath{\mathcal{A}}\xspace}

\newcommand{\hist}{\ensuremath{\mathsf{H}}\xspace}
\newcommand{\nhist}{\ensuremath{\tilde{\mathsf{H}}}\xspace}

\newcommand{\nhc}{\ensuremath{\mathit{HC}^{\tilde{\mathsf{H}}(D,g)}}\xspace}

\newcommand{\pfc}{PrivPfC\xspace}
\newcommand{\dgen}{DiffGen\xspace}
\newcommand{\pbayes}{PrivBayes\xspace}
\newcommand{\pph}{PPH\xspace}
\newcommand{\pgene}{PrivGene\xspace}
\newcommand{\dpc}{DiffPC-4.5\xspace}
\newcommand{\perm}{PrivateERM\xspace}
\newcommand{\pview}{PriView\xspace}

\newcommand{\pfcSelNF}{PrivPfC-SelNF\xspace}
\newcommand{\pfcFSNF}{PrivPfC-FSNF\xspace}
\newcommand{\dgenSNF}{DiffGen-Struct-NF\xspace}
\newcommand{\dgenPNF}{DiffGen-Pert-NF\xspace}
\newcommand{\dgenNF}{DiffGen-NF\xspace}
\newcommand{\pbayesSNF}{PrivBayes-Struct-NF\xspace}
\newcommand{\pbayesPNF}{PrivBayes-Pert-NF\xspace}
\newcommand{\pbayesNF}{PrivBayes-NF\xspace}

\allowdisplaybreaks

\begin{document}
\title{Differentially Private Projected Histograms of Multi-Attribute Data for Classification}

\author{%
	{Dong Su$^{~\#}$, Jianneng Cao$^{~\star}$, Ninghui Li$^{~\#}$}%
	\vspace{1.6mm}\\
	\fontsize{10}{10}\selectfont\itshape
	$^{\#}$\,
    Purdue University \hspace{3cm}$^{\star}$\,Institute for Infocomm Research\\
	\fontsize{9}{9}\selectfont\ttfamily\upshape
	\hspace{-1cm}\{su17, ninghui\}@cs.purdue.edu \hspace{1.5cm}caojn@i2r.a-star.edu.sg%
}

\maketitle

\begin{abstract}
\input{abstract}
\end{abstract}

\section{Introduction}
\input{introduction}

\section{Related Work}\label{sec:related}
\input{related.tex}

\section{\pfc Framework}
\input{single_grid.tex}

\section{Experiment}\label{sec:expt}
\input{experiment.tex}


\section{Conclusion}
\input{conclusion.tex}


\bibliographystyle{abbrv}
{\bibliography{privacy}}

\section{Appendix}\label{sec:appen}
\input{appendix.tex}

\end{document}

%% file: abstract.tex
In this paper, we tackle the problem of constructing a differentially private synopsis for the classification analyses.  Several the state-of-the-art methods follow the structure of existing classification algorithms and are all iterative, which is suboptimal due to the locally optimal choices and the over-divided privacy budget among many sequentially composed steps.  Instead, we propose a new approach, \pfc, a new differentially private method for releasing data for classification.  The key idea is to privately select an optimal partition of the underlying dataset using the given privacy budget in one step.  Given one dataset and the privacy budget, \pfc constructs a pool of candidate grids where the number of cells of each grid is under a data-aware and privacy-budget-aware threshold.  After that, \pfc selects an optimal grid via the exponential mechanism by using a novel quality function which minimizes the expected number of misclassified records on which a histogram classifier is constructed using the published grid.  Finally, \pfc injects noise into each cell of the selected grid and releases the noisy grid as the private synopsis of the data.   If the size of the candidate grid pool is larger than the processing capability threshold set by the data curator,  we add a step in the beginning of \pfc to prune the set of attributes privately.  We introduce a modified $\chi^2$ quality function with low sensitivity and use it to evaluate an attribute's relevance to the classification label variable.  Through extensive experiments on real datasets, we demonstrate \pfc's superiority over the state-of-the-art methods. 

%% file: introduction.tex

We study the problem of publishing histograms of datasets while satisfying differential privacy.  
A histogram is an important tool for summarizing data, and can serve as the basis for many data analysis tasks.  Publishing noisy histograms for one-dimensional or two-dimensional datasets have been studied extensively in recent years~\cite{HRMS10,XXY10, CPS+11, Lei11, CSS11,XWG11,Vinterbo12, CPS+12, QYL13a, QYL13b}.  However, as noticed in~\cite{Vinterbo12,QYL13a}, these approaches do not work well when the number of attributes/dimensions goes above a few.  Many datasets that are of interest have multiple attributes.  In this paper, we focus on multi-attribute datasets that have dozens of attributes, some of categorical and some are numerical.

For such a multi-attribute dataset, it is infeasible to publish a histogram with all the attributes, therefore it is necessary to select a subset of the attributes that are ``interesting'' for some intended data analysis tasks, and to determine how to discretize the attributes.  These selections partition the domain into a number of cells.  We call the result a ``grid''.  
We consider a common optimization objective, where the dataset includes a label attribute and our goal is to ensure that classifiers that are accurate for the original dataset can be learnt from the published noisy histograms.  Classification is an important tool for data analysis, and differentially private algorithms for learning classifiers have been considered an important problem, with many recent attempts~\cite{BDMN05, CM08, FS10, MCFY11, CMS11, JJW+12, JPW12, Vinterbo12, ZXY+13, ZCP+14}. 
In this paper we propose the \pfc (Private Publication for Classification) approach for publishing projected histograms.  The key novelty is to privately select a high-quality grid in a single step, while adapting to the privacy budget $\epsilon$.  We construct a set of candidate grids where the number of cells is under a certain threshold (determined by the dataset size and the privacy budget), and then use the exponential mechanism to select one grid using a novel quality function that minimizes expected number of misclassified records when a histogram classifier is constructed using the published histogram.  
By construction, our quality function considers the impact of injected Laplace noise to the histogram on the classification accuracy.

For high dimensional datasets, the size of the set of candidate grids might be larger than the processing capacity of the data curator.  We add a feature selection step in the beginning to prune the set of attributes.  This step first privately selects a small number of attributes which are most relevant to the class attribute by employing the exponential mechanism.  
We introduce a modified $\chi^2$ correlation function that has low sensitivity while evaluating an attribute's relevance to the classification label variable.  This feature selection step enables our \pfc framework to scale to higher dimensional datasets.


Our proposed \pfc outputs a histogram that can be used to generate synthetic data for multiple data analysis tasks, while being optimized for data classification.  We show the effectiveness of \pfc by comparing it with several other approaches that output a classifier in a differentially private fashion.

For evaluation, we use two common classification algorithms, the decision tree and the SVM, because these have been used in the literature on learning classifiers while satisfying the differential privacy.  Extensive experiments on real datasets show that \pfc consistently and significantly outperforms other state-of-the-art methods. 

The contributions of this paper are summarized as follows:
\begin{enumerate}
\item We propose \pfc, a novel framework for publishing data for classification under differential privacy.  As part of \pfc, we introduce a new quality function that enables the selection of a good ``grid'' for publishing noisy histograms.  We also introduce a way to enable private selection of most relevant features for classification, and to enable \pfc to scale to higher-dimension datasets.

\item Through extensive experiments on real datasets, we have compared \pfc against several other state-of-the-art methods for data publishing as well as private classification, demonstrating that \pfc improves the state-of-the-art.  
\end{enumerate}

The rest of the paper is organized as follows. In Section 2, we review the related work.  Our \pfc approach is presented in Section 3.  We report experimental results in Section 4.  Section 5 concludes our work.

%% file: related.tex
The notion of differential privacy was developed in a series of papers \cite{DN03,DN04,BDMN05,DMNS06,Dwo06}.  
There are several primitives for satisfying $\epsilon$-differential privacy.  In this paper we use two of them.  The first primitive is the \emph{Laplacian mechanism}~\cite{DMNS06}.  It adds noise sampled from a Laplace distribution to a statistic $f$ to be released.  The scale of the Laplace distribution is proportional to $\mathsf{GS}_f$, the \emph{global sensitivity} or the $L_1$ sensitivity of $f$. 
Another primitive is to sample the output of the data analysis mechanism according to an exponential distribution; this is generally referred to as the \emph{exponential mechanism}~\cite{MT07}.  The mechanism relies on a \emph{quality} function $q: \mathcal{D} \times \mathcal{R} \rightarrow \mathbb{R}$ that assigns a real valued score to one output $r\in\mathcal{R}$ when the input dataset is $D$, where higher scores indicate more desirable outputs.  Given the quality function $q$, its global sensitivity $\mathsf{GS}_q$ is defined as:
\begin{equation*}
\mathsf{GS}_q = \max_{r} \max_{D\simeq D'} |q(D,r) - q(D',r)|.
\end{equation*}

The following method satisfies $\epsilon$-differential privacy:
\begin{equation*}
 \Pr{r\mbox{ is selected}} \propto \myexp{\left(\frac{\epsilon }{2\,\mathsf{GS}_q}q(D,r)\right)}.  \label{eq:exp}
\end{equation*}

There has been a large body of works on differentially private histogram construction for answering range queries or marginal queries~\cite{DMNS06, XWG11, HRMS10, CPS+12, XuZXYY12, QYL13a, LHMW14, QYL14, ZCP+14}. 

\mypara{Differentially Private Classification.} Differentially private classification has received growing attention in the research community~\cite{BDMN05, CM08, FS10, MCFY11, CMS11, JJW+12, JPW12, Vinterbo12, ZXY+13, ZCP+14}.  Blum et al.~\cite{BDMN05} suggested a solution for constructing the private version of the ID3~\cite{Quinlan86} decision tree classifier.  When the ID3 algorithm needs to get the number of tuples with a specific feature value, it queries the SuLQ interface to get the corresponding noise count.  Friedman and Schuster~\cite{FS10} improved this approach by redesigning the classic ID3 classifier construction algorithm to consider the feature quality function with low sensitivity and using exponential mechanism to evaluate all the attributes simultaneously.  Chaudhuri et al.~\cite{CM08, CMS11} proposed a differentially private logistic regression algorithm and later generalized this idea to address the private empirical risk minimization which can be applied to a wider range of classification problems, such as SVM classification.  Zhang et al.~\cite{ZXY+13} proposed \pgene, a general private model fitting framework based on genetic algorithms, that can be applied to the SVM classification and the logistic regression.  

Besides the above interactive methods for constructing differentially private classifiers, several works proposed solutions to publish data for classification analysis tasks.  Mohammed et al.~\cite{MCFY11} proposed the \dgen algorithm which first partitions the data domain by iteratively selecting attributes and ways to discretize the attributes,  and then injects Laplace noise into each cell of all the leaf partitions.  Vinterbo~\cite{Vinterbo12} proposed another data publishing algorithm, called Private Projected Histogram (\pph).  \pph first decides how many attributes are to be selected, then incrementally selects attributes via the exponential mechanism to maximize the discernibility of the selected attributes.  For each categorical attribute, the full domain is used.  For numerical attribute, it uses the formula proposed in Lei~\cite{Lei11} to decide how many bins to discretize them.  In this method, the number of attributes and how attributes are partitioned are independent of the privacy budget.  Furthermore, all selected attributes are treated equally.  Zhang et al.~\cite{ZCP+14} presented \pbayes which constructs a private a Bayesian network through iteratively selecting sets of attributes that have maximum mutual information via the exponential mechanism.  It then injects Laplace noise to perturb each conditional distribution of the network.  We will further analyze the above approaches and compare our proposed method with them in the later sections.

%% file: single_grid.tex
In this section we present the \pfc framework of privately publishing data for classification analysis.

\subsection{Preliminaries}

\begin{figure*}[!htb]
\centering
\includegraphics[scale=0.8]{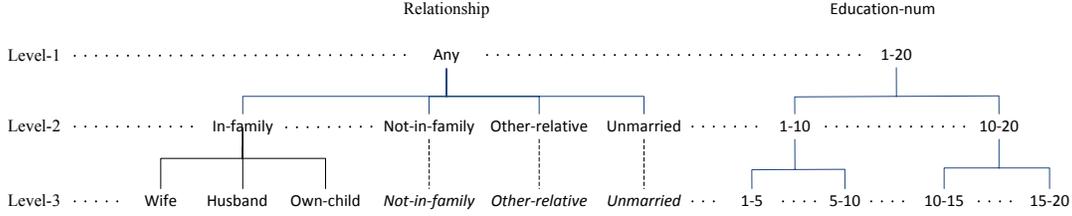}
\caption{Taxonomy hierarchies of Relationship attribute and Education-num attribute}\label{fig:taxo-hier}
\end{figure*}

We consider a dataset with a set of predictor variables and one binary response variable. The predictor variables can be numerical or categorical. 
Following~\cite{BA05, Iye02, FWY05, MCFY11}, for each predictor variable $A_i$, we assume the existence of a \emph{taxonomy hierarchy} (also called a \emph{generalization hierarchy} in the literature) $T_i$. 
Figure~\ref{fig:taxo-hier} shows the taxonomy hierarchies of \emph{Relationship}, a categorical variable, and \emph{Education-num}, a numerical variable. 
In the hierarchy, the root node represents the whole domain of the variable, and a parent node is a generalization (or a cover) of its children. Child nodes under the same parent node are semantically related; they are closer to each other than to nodes under a different parent node.

Each \emph{level} of a predictor variable's taxonomy hierarchy forms a partition of its domain. On the basis of the taxonomy hierarchy and its levels, we introduce the notion of a grid.

\begin{definition}[Grid]\label{def:grid}
Let $\mathsf{A} = \{A_1, \ldots, A_d\}$ be the set of predictor variables in a dataset and $\{T_{1},\ldots, T_{d}\}$ be their taxonomy hierarchies respectively.  Let $h_{i}$ be the height of $T_{i}$, $1\leq i \leq d$.  Then, a \emph{grid} $g$ is given by $\langle \ell_{1},\ldots, \ell_{d}\rangle$, where $1 \leq \ell_{i}\leq h_{i}$ and $1\leq i \leq d$.  A grid defines a partition of the data domain into \emph{cells} where each attribute $A_i$ is partitioned into the values at level $\ell_i$.
The number of cells of a grid is $\Pi_{i=1}^{d}|T_{i}[l_{i}]|$, where $|T_{i}[l_{i}]|$ is the number of nodes in the level $l_{i}$ of the hierarchy $T_i$.  And the number of all possible grids is $\Pi_{i=1}^{d} h_{i}$.
\end{definition}

\begin{definition}[Histogram]\label{def:histogram}
Given a dataset $D$ and a grid $g$, a histogram $\hist(D, g)$ partitions $D$ into cells according to $g$, and outputs the numbers of positive instances and negative instances in each cell.
\end{definition}
By injecting Laplace noise into the positive counts and negative counts of each cell in the histogram $\hist(D, g)$, we get the noisy version of it, $\nhist(D, g)$.

\subsection{Histogram Publishing for Classification}\label{sec:histogramPublishing}
Given a dataset $D$, the taxonomy hierarchies of its predictor variables, a total privacy budget $\epsilon$, and the number of tuples in the dataset $N$ (a rough estimate suffices), we generate a candidate pool of all grids whose number of cells are below a threshold, which is determined by $\epsilon$ and $N$.  We compute the quality score for each grid, which measures the usefulness for classification of each grid in the pool.  We then apply the exponential mechanism~\cite{MT07} to privately select a grid, and finally publish a noisy histogram using $g$. 

A key technical challenge is to come up with a low-sensitivity quality function that can measure the desirability of choosing a particular grid $g$.  We publish $\tilde{\mathsf{H}}(D,g)$, a noisy histogram of $D$ using $g$ to partition the data domain, and desire that classifiers learned from $\tilde{\mathsf{H}}(D,g)$ are close to classifiers learned from $D$.  Furthermore, we desire this to hold regardless of which particular classification algorithm is used.  We propose to define the quality function to minimize the misclassification error (when measured using the dataset $D$) for the classifier defined by the histogram $\tilde{\mathsf{H}}(D,g)$, i.e., for each cell in the grid defined by $g$, it predicts the majority class according to $\tilde{\mathsf{H}}(D,g)$.  This classifier is in the same spirit as histogram classifiers~\cite{DGL96, Nowak09}, and we use \nhc to denote it.



Suppose that a grid is able to separate positive and negative data points very well, then, even after adding the noises, this separation feature is still preserved and can be used to learn classifiers.  When no noise is added, the finest partition is desired.  However, with noise, we want to ensure that the noises do not overwhelm the true counts.
Since when a grid $g$ is fixed, the noisy histogram includes random noises, the misclassification error is a random variable, and we use the expected value of this error as the quality function.

\begin{definition}[Quality of grid]
Given a dataset $D$ and a grid $g$, the grid quality is measured by the expected misclassification error of the histogram classifier \nhc:
$$\mathrm{qual}_{D}(g) =  \mathbb{E}[err(\nhc, D)].$$
\end{definition}

The following Lemma shows how to compute $\mathrm{qual}_{D}(g)$.

\begin{lemma}[Quality of grid] \label{lemma:computeQual}
Given a dataset $D$ and a grid $g$, $\epsilon$ for the parameter of adding Laplacian noise to the counts, we have
\begin{align}
\mathrm{qual}_{D}(g) &= \sum_{c\in g} \left [\min(n_{c}^{+}, n_{c}^{-}) \left( 1 - \frac{e^{-\epsilon x_c}}{2}\left( 1 + \frac{\epsilon x_c}{2}\right) \right) \right. \nonumber \\  &+ \left. \max(n_{c}^{+}, n_{c}^{-}) \left( \frac{e^{-\epsilon x_c}}{2}\left( 1 + \frac{\epsilon x_c}{2}\right) \right)\right],\label{eqn:qualHWithNoise}
\end{align}
where $c$ ranges over all cells in the grid, $n_{c}^{+}$ is the number of positive data points in $c$, $n_{c}^{-}$ is the number of negative data points in $c$, and $x_c = |n_{c}^{+} - n_{c}^{-}|$.
\end{lemma}

To prove Lemma~\ref{lemma:computeQual}, we note that $\mathrm{qual}_{D}(g)$ can be further decomposed into the sum of expected misclassification error at each perturbed cell of the histogram after majority voting, and thus
\begin{align*}
\mathrm{qual}_{D}(g) 
                    &= \sum_{c\in g} \mathbb{E}[err(\tilde{c}, D)]
\end{align*}
Where $\tilde{c}$ denotes the application of the histogram classifier \nhc to the cell $c$.

For cell $c$, if the added Laplace noises do not change the majority class label, then the number of misclassified input tuples is $\min(n_{c}^{+}, n_{c}^{-})$; otherwise, it is $\max(n_{c}^{+}, n_{c}^{-})$.  Thus,
\begin{align}\label{eqn:noisyRiskForCell}
\mathbb{E}[err(\tilde{c}, D)] = \min(n_{c}^{+}, n_{c}^{-})\cdot p_c + \max(n_{c}^{+}, n_{c}^{-})\cdot (1 - p_c),
\end{align}
where $p_c$ is the probability that the majority class label in $c$ does not change after injecting Laplace noises.


Let $Z_1$ and $Z_2$ be the Laplace noises added to the majority class and the minority class of cell $c$, respectively, then
\begin{equation}\label{eqn:z2Minusz1}
p_c = \Pr{Z_{2} - Z_{1} \leq |n_{c}^{+} - n_{c}^{-}| }.
\end{equation}

\begin{lemma}[\cite{kotz2001laplace}]\label{lem:z1Minusz2}
Let $Z_1$ and $Z_2$ be two i.i.d. random variables that follow the Laplace distribution with mean 0 and scale $\frac{1}{\epsilon}$. Then the density of their difference $Y = Z_1 - Z_2$ is
\begin{align*}
f_{Y}(y) = \frac{\epsilon}{4} e^{-\epsilon |y|}(1 + \epsilon|y|)\quad\quad   -\infty < y < \infty,
\end{align*}
and the corresponding cumulative distribution function is
\begin{align}\label{eqn:z1Minusz2CDF}
	F_{Y}(y) =
	\begin{dcases}
		1 - \frac{e^{-\epsilon y}}{2}\left( 1 + \frac{\epsilon y}{2}\right),  &  \mathrm{if~} y \geq 0,\\
		\frac{e^{\epsilon y}}{2}\left( 1 - \frac{\epsilon y}{2}\right), & \mathrm{otherwise}.
	\end{dcases}
\end{align}
\end{lemma}

From Equations (\ref{eqn:z2Minusz1}) and (\ref{eqn:z1Minusz2CDF}), we have
\begin{equation}\label{eqn:pcm}
p_c = 1 - \frac{e^{-\epsilon |n_{c}^{+} - n_{c}^{-}|}}{2}\left( 1 + \frac{\epsilon |n_{c}^{+} - n_{c}^{-}|}{2}\right).
\end{equation}
Plugging Equation (\ref{eqn:pcm}) into Equation (\ref{eqn:noisyRiskForCell}) provides Lemma~\ref{lemma:computeQual}.

The lemma below bounds the sensitivity of our quality function.

\begin{lemma}\label{lemma:gridQualSen}
For any $\epsilon > 0$, the global sensitivity of the quality function~\ref{eqn:qualHWithNoise} is $B(\epsilon)$, where
        \begin{align*}
            B(\epsilon) &= x \cdot\left(\frac{e^{-\epsilon(x-1)}}{2}\left(1+\frac{\epsilon(x-1)}{2}\right) - \frac{e^{-\epsilon x}}{2}\left(1+\frac{\epsilon x}{2}\right)\right)\\
                 &+ \left(1 - \frac{e^{-\epsilon(x-1)}}{2}\left(1+\frac{\epsilon (x-1)}{2}\right)\right),
        \end{align*}
and
        \begin{align*}
            x = \frac{\epsilon e^{\epsilon} + \sqrt{2 - \left(4 - \epsilon^{2}\right)e^{\epsilon} + 2e^{2\epsilon} }}{-\epsilon + \epsilon e^{\epsilon}}.
        \end{align*}
\end{lemma}

\begin{algorithm}
	\caption{\pfc: Privately Publishing Data for Classification}\label{alg:all-in-one}
	\textbf{Input: } dataset $D$, the set of predictor variables $\mathcal{F}$ and their taxonomy hierarchies, total privacy budget $\epsilon$, maximum grid pool size $B$, median of the first branching factors of hierarchies $b$.\\
	
	\begin{algorithmic}[1]

        \Function{$\mathsf{main}$}{$D, \mathcal{F}, N, \epsilon$, $B$, $b$}
            \State $T \leftarrow \delta\cdot N\cdot\epsilon/2$ \label{line:T}
			\State $\hist \leftarrow \mbox{Enumerate}(\mathcal{F}, T)$\label{line:basicAlgoEnumeration}
            \If {$|\hist| \geq B$}
                \State $k\leftarrow \left\lceil \frac{2\log{T}}{\log{b}} \right\rceil$
            	\State $\epsilon_{fs} \leftarrow 0.3\epsilon$, $\epsilon_{sh} \leftarrow 0.3\epsilon$, $\epsilon_{ph} \leftarrow 0.4\epsilon$
	        	\State $X \leftarrow \FS(D, \mathcal{F}, k, \epsilon_{fs})$ \label{line:featureSelection}
    			\State $\hist_{X} \leftarrow \mbox{Enumerate}(X, T)$\label{line:basicAlgoEnumeration}
			    \State $\hat{I} \leftarrow \mbox{PrivateHistogramPublishing}(D, \hist_{X}, \epsilon_{sh}, \epsilon_{ph})$ \label{line:allPSPT}
            \Else
                \State $\epsilon_{sh} \leftarrow \frac{3}{7}\epsilon$, $\epsilon_{ph} \leftarrow \frac{4}{7}\epsilon$
			    \State $\hat{I} \leftarrow \mbox{PrivateHistogramPublishing}(D, \hist, \epsilon_{sh}, \epsilon_{ph})$ \label{line:allPSPT}
            \EndIf
            \State\Return $\hat{I}$
        \EndFunction

        \Statex

        \Function{$\mathsf{PrivateHistogramPublishing}$}{$D, \hist, \epsilon_{sh}, \epsilon_{ph}$}\label{func:histpub}
			\State $h \leftarrow \HS(D, \hist, \epsilon_{sh})$
			\State $\hat{I} \leftarrow \HP(D, h, \epsilon_{ph})$
            \State\Return $\hat{I}$
        \EndFunction

        \Statex
        \Function{\HS}{$D, \hist, \epsilon_{sh}$}
			\For {$i = 1 \to |\hist|$}
				\State $q_i \leftarrow \mathrm{qual}(\hist_i) $ 
				\State $p_i \leftarrow \myexp{-(q_i\epsilon_{sh}) / 2}$
			\EndFor
			\State $h \leftarrow$ sample $i \in [1..|\hist|]$ according to $p_i$
            \State\Return $h$
        \EndFunction

        \Statex
        \Function{\HP}{$D, h, \epsilon_{ph}$}
            \State Initialize $I$ to empty
			\For {each cell $c \in h$}
				\State $\hat{n}_{c}^{+} \leftarrow n_{c}^{+} + \mathrm{Lap}(1/\epsilon_{ph})$
				\State $\hat{n}_{c}^{-} \leftarrow n_{c}^{-} + \mathrm{Lap}(1/\epsilon_{ph})$
                \State Add $(\hat{n}_{c}^{+}, \hat{n}_{c}^{-})$ to $I$
			\EndFor
			\State Round all counts of $I$ to their nearest non-negative integers.
            \State\Return $I$
        \EndFunction
        \item[]

        \Statex
        \Function{\FS}{$D, \mathcal{F}, k, \epsilon_{fs}$}
            \State Initialize $X$ to empty
            \State Let $R$ be the response variable in $D$
            \For {{\bf each} $A_i \in \mathcal{F}$}
                \State $cor_i \leftarrow \cor{A_i, R, D}$
                \State $p_i \leftarrow e^{\frac{\epsilon_{fs}\cdot cor_i}{4k}}$
            \EndFor
            \For {$i = 1 \to k$}
		        \State $f \leftarrow $sample $A_i\in \mathcal{F}$ according to $p_i$
		        \State Add $f$ to $X$
		        \State Remove $f$ from $\mathcal{F}$
            \EndFor
            \State \Return $X$
        \EndFunction
	\end{algorithmic}
\end{algorithm}

\subsection{Correlation-based Feature Selection}\label{sec:featureSelection}

Our basic solution privately selects a grid from the candidate pool to release synthetic data.  As the number of predictor variables increases, the candidate pool size grows exponentially, giving arise to a scalability issue.  Fortunately, usually in a real dataset some predictor variables are not useful for predicting the class labels.  Such irrelevant variables can thus be excluded from the classification to improve the scalability of our solution.

Feature selection~\cite{IE03} is the process of selecting a subset of important features (predictor variables) to build a classification model.
Various feature selection methods have been proposed including wrapper method, embedded methods, stepwise regression~\cite{Efroymson}. However, they require building a large number of classification models, one for a subset of features one wants to evaluate.  It is unclear how to adapt these methods to satisfy differential privacy.
We propose a simple but effective approach, which selects predictor variables based on a correlation analysis between predictor variables and the class (i.e., target variable).  We adapt the \chitest~\cite{HK11} to have a low sensitivity.


Given a dataset $D$ with $N$ tuples, the \chitest~\cite{HK11} evaluates whether categorical variables $A$ and $B$ are correlated. Suppose that variable $A$ has $m$ distinct values, $a_1, \ldots, a_m$, and $B$ has $n$ distinct values, $b_1, \ldots, b_n$. The $\chi^2$ correlation between $A$ and $B$ (a.k.a Pearson $\chi^2$ statistic) is defined as

\begin{equation}\label{eqn:chiSquare}
\chi^2(A, B) = \sum_{i=1}^{m} \sum_{j = 1}^{n} \frac{(o_{ij} - e_{ij})^2}{e_{ij}},
\end{equation}
where $o_{ij}$ is the observed number of tuples with $A = a_i$ and $B = b_j$ in dataset $D$, and the expected count $e_{ij}$ is computed by assuming $A$ and $B$ are independent
\begin{equation}\label{eqn:expectedFrequency}
e_{ij} = \frac{\mathrm{count}(A = a_i) \times \mathrm{count}(B = b_j)}{N},
\end{equation}
where $\mathrm{count}(A = a_i)$ returns the number of tuples in dataset $D$ with $A=a_i$.

Clearly, if $A$ and $B$ are independent, then $\chi^2(A, B) = 0$. 
The bigger the $\chi^2$ value is, the stronger the correlation of variables $A$ and $B$ is.  $\chi^2(A, B)$, however, has a large global sensitivity, because the $e_{ij}$'s in Equation~(\ref{eqn:chiSquare}) can be very small.  Our analysis (omitted for space limitation) shows that it is at least $\frac{N^2 + 1}{2N}$.

We adapt \chitest to define the correlation $\mathsf{Cor}(A, R)$ between a predictor variable $A$ with $m$ distinct values and the binary response variable $R$ as:
\begin{equation}\label{eq:feature-correlation-function}
\cor{A, R} = \sum_{i=1}^{m} \sum_{j = 1}^{2} \left |o_{ij} - e_{ij}\right |,
\end{equation}
where $o_{ij}$ is the observed number of tuples with $A = a_i$ and $R = r_j$ in $D$, and the expected count $e_{ij}$ is computed by assuming $A$ and $R$ are independent as in Equation~(\ref{eqn:expectedFrequency}).




\begin{lemma}\label{lem:corr-sen}
Let $A$ and $R$ be a categorical variable and binary response variable in a dataset $D$, respectively. Then, the global sensitivity of Function \cor{A, R, D} is 2.
\end{lemma}

The proof is in the Appendix Section.

\subsection{The Algorithm}\label{sec:wholeFramework}

We now present the full algorithm (Algorithm~\ref{alg:all-in-one}) for our framework of releasing private data for classification tasks.

Line \ref{line:T} sets the threshold of the maximum number of cells in a grid, to prevent the average counts from being dominated by the injected noises.  That is,
\begin{equation}\label{eqn:nonDominatingRuleTmp}
E\left[\left|\Lap{\frac{1}{\epsilon}}\right|\right]\leq \frac{1}{5}\cdot\frac{N}{T},
\end{equation}
which means that the average noise magnitude is no more than the 20\% of the average cell count.

When feature selection is deemed necessary, we allocate 30\% of the privacy budget to privately select $k$ predictor variables that are strongly correlated with the response variable.  These selected variables are then used to release private synthetic data.

The number of attributes to be selected, $k$, is based on $T$, the maximum grid size.  We want to have enough attributes so that relevant attributes are included.  At the same time, we do not want $k$ so large so that there are too many candidate grids with size below $T$.  Let $b$ be the median of the first branching factors of hierarchies of all attributes. We set $k$ to be $\frac{2\log(T)}{\log(b)}$.

\begin{theorem}\label{thm:extendedAlgoDP}
Algorithm~\ref{alg:all-in-one} satisfies $\epsilon$-differential privacy.
\end{theorem}

Theorem \ref{thm:extendedAlgoDP} shows that Algorithm \pfc satisfies \linebreak $\epsilon$-\difp. The proof of Theorem \ref{thm:extendedAlgoDP} is thus straightforward by considering the sequential composability of \difp as discussed in Section \ref{sec:related}.

%% file: experiment.tex
\begin{table*}
\centering
\begin{tabular}{|c|c|c|c|c|c|}\hline
{\bf Dataset} & {\bf \# Dim} & {\bf \# Numerical} & {\bf \# Categorical} & {\bf \# Records} & {\bf Classification Task}\\ \hline
Adult & 15 & 6 & 8 & 45,222 & Determine whether a person makes over 50K a year.\\ \hline
Bank & 21 & 10 & 10 & 41,188 & Determine whether the client subscribed a term deposit.\\ \hline
US & 47 & 15 & 31 & 39,186 & Determine whether a person makes over 50K a year.\\ \hline
BR & 43 & 14 & 28 & 57,333 & Determine whether a person makes over 300 per month.\\ \hline
\end{tabular}
\caption{Dataset characteristics}\label{tab:dataset_desc}
\end{table*}

\begin{table*}
\centering
\begin{tabular}{|c|c|c|}\hline
\quad & {\bf Methods} & {\bf Description}\\ \hline
\multirow{8}{*}{Non-Interactive}
&\pfc & Our proposed method. \\ \cline{2-3}
&\pfcSelNF & Our proposed method with noise free feature selection and histogram selection. \\ \cline{2-3}
&\pfcFSNF & Our proposed method with noise free feature selection. \\ \cline{2-3}
&\dgen~\cite{MCFY11} & Private data release for classification via recursive partitioning. \\ \cline{2-3}
&\dgenNF & Noise free \dgen. \\ \cline{2-3}
&\pbayes~\cite{ZCP+14} & Private Data Release via Bayes network. \\ \cline{2-3}
&\pbayesNF & Noise free \pbayes. \\ \cline{2-3}
&\pph~\cite{Vinterbo12} & Private data release for classification by projection and perturbation.\\ \hline\hline
\multirow{3}{*}{Interactive}
&\dpc~\cite{FS10} & Privately construct C4.5 decision tree classifier. \\ \cline{2-3}
&\pgene~\cite{ZXY+13} & Private model fitting based on genetic algorithms.\\ \cline{2-3}
&\perm~\cite{CMS11} & Private classifier construction based on empirical risk minimization. \\ \hline
\end{tabular}
\caption{Summary of differentially private classification methods}\label{tab:methods_summary}
\end{table*}

\begin{figure*}[!htb]
\begin{tabular}{cc}

   \multicolumn{2}{c}{\hspace{-2mm}\includegraphics[width = \textwidth]{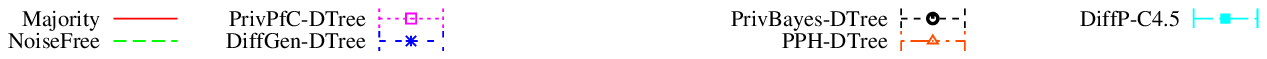}}\\
	\includegraphics[width = 3.0in]{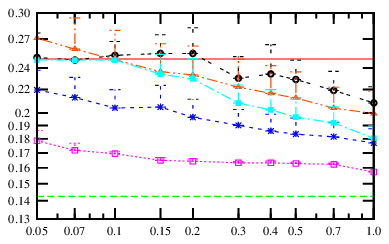}
	& \includegraphics[width = 3.0in]{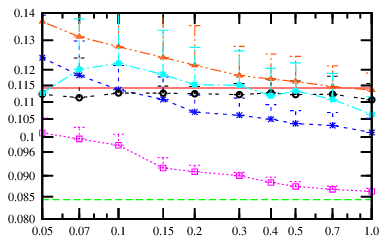}\\
	(a) Adult & (b) Bank\\
	\includegraphics[width = 3.0in]{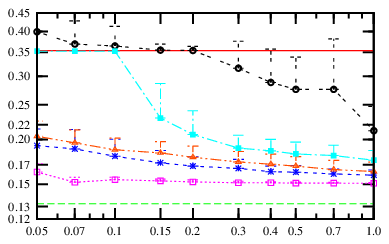}
	& \includegraphics[width = 3.0in]{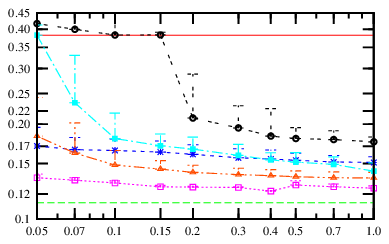} \\
	(c) US & (d) BR\\

\end{tabular}
	\caption{Comparison of \pfc, \dgen, \pbayes, \pph and \dpc by decision tree classification.  x-axis: privacy budget $\epsilon$ in log-scale. y-axis: misclassification rate in log-scale.}\label{fig:comparison-dtree}
\end{figure*}

\begin{figure*}[!htb]
\begin{tabular}{cc}


    \multicolumn{2}{c}{\hspace{-2mm}\includegraphics[width = \textwidth]{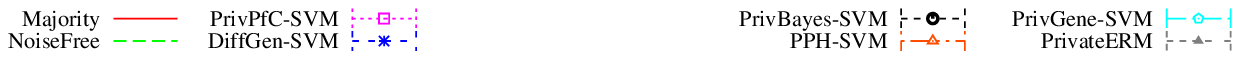}}\\
	\includegraphics[width = 3.0in]{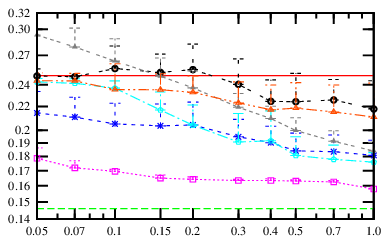}
	& \includegraphics[width = 3.0in]{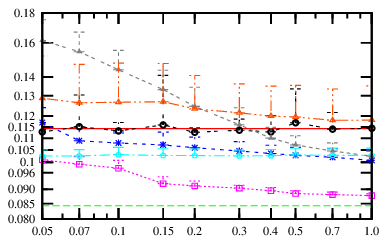}\\
	(a) Adult & (b) Bank\\
	\includegraphics[width = 3.0in]{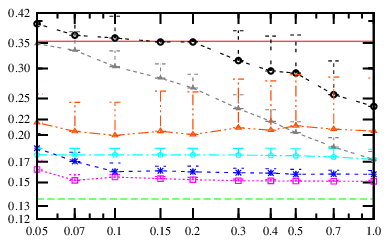}
	& \includegraphics[width = 3.0in]{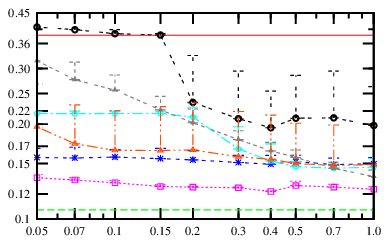} \\
	(c) US & (d) BR\\

\end{tabular}
	\caption{Comparison of \pfc, \dgen, \pbayes, \pph, \pgene and \perm by SVM classification. x-axis: privacy budget $\epsilon$ in log-scale. y-axis: misclassification rate in log-scale.}\label{fig:comparison-svm}
\end{figure*}

\subsection{Experimental Settings}
\mypara{Datasets.} We use 4 real datasets for our experiments. The first one is the Adult dataset from the UCI machine learning repository~\cite{AN10}.
It contains 6 numerical attributes and 8 categorical attributes, and is widely used for evaluating the performance of classification algorithms. After removing missing values, the dataset contains 45,222 tuples.
The second dataset is the Bank marketing dataset from the same repository. It contains 10 numerical attributes and 10 categorical attributes on 41,188 individuals. The third is the US dataset from the \emph{Integrated Public Use Microdata Series} (IPUMS)~\cite{ipums}. It has 39,186 the United States census records in 2010, with 15 numerical attributes and 31 categorical ones. The last is the BR dataset (also from IMPUS), which contains 57,333 Brazil census records in 2010 and has 14 numerical attributes and 28 categorical ones.
The classification tasks for the Adult, US and BR datasets are to predict whether an individual has an income above a certain threshold.  The one for the Bank dataset is to predict whether a client will subscribe a term deposit. Table~\ref{tab:dataset_desc} summarizes the characteristics of the datasets.

\mypara{Taxonomy Hierarchies.}
For the Adult dataset, we use the same taxonomy hierarchies as \dgen~\cite{MCFY11}.  For the remaining 3 datasets, we do the following.  For numerical attributes,
we partition each domain into equal size bins and build hierarchies over them.  For categorical attributes, we build taxonomy hierarchies by considering the semantic meanings of the attribute values.

\mypara{Competing Methods.} We compare \pfc with 6 state-of-the-art methods in terms of misclassification rate.  These include 3 non-interactive methods, \dgen~\cite{MCFY11}, \pbayes~\cite{ZCP+14} and Private Projected Histogram (\pph)~\cite{Vinterbo12},  which privately release synthetic datasets for classification analyses, and 3 interactive methods, \pgene~\cite{ZXY+13}, \dpc~\cite{FS10}, and \perm~\cite{CMS11}, which includes one method for decision tree, and two methods for SVM.

\myspara{\dgen.}~\cite{MCFY11} consists of two steps, partition and perturbation.  The partition step first generalizes all attribute's values into the topmost nodes in their taxonomy hierarchies and then iteratively selects one attribute at a time for specialization, using the exponential mechanism.  The quality of each candidate specialization is based on the same heuristics as used by the decision tree algorithms, such as information gain and majority class.  As suggested in~\cite{MCFY11}, we use the majority class to measure the candidate quality, and set the number of specialization steps to be 10 for the Adult dataset and the bank dataset.  For the US and BR datasets, we set the number to be 6 and 8 respectively, as beyond these numbers, the \dgen implementation runs into memory problems.  The perturbation step injects Laplace noise into each cell of the partition and outputs all the cells with their noisy counts as the noisy synopsis of the data.

\myspara{\pbayes.}~\cite{ZCP+14} determines the structure of a Bayesian network by first randomly select an attribute as the first node, and then iteratively select one attribute and up to $k$ nodes as the attribute's parent nodes, which have the maximum mutual information.  After the structure is determined, \pbayes perturbs the marginals needed for computing the conditional distributions.  The performance of the \pbayes algorithm depends on $k$.  We set $k=3$ for the Adult dataset and the Bank dataset, which is the same as the one used in~\cite{ZCP+14}.  For the US and BR datasets, which were not used in~\cite{ZCP+14}, setting $k=3$ runs out of memory in our experiments because of the larger dimensionality; we set $k=2$ for them.


\myspara{\pph.}~\cite{Vinterbo12} starts with a feature selection procedure to select a set of $k$ features that have the maximal discernibility.  Then, it uses the selected features to build a histogram.  For each categorical attribute, the full domain is used.  For numerical attribute, it uses the formula proposed in Lei~\cite{Lei11} to decide how many bins to discretize them.

\emph{\pgene}~\cite{ZXY+13} is a general-purpose private model fitting framework based on genetic algorithms, which can be applied to SVM classification.  \emph{\dpc}~\cite{FS10} is an interactive private algorithm for building a C4.5 decision tree classifier differential-privately.  \emph{\perm}~\cite{CMS11} is an interactive private algorithm for constructing SVM classifier by injecting noise into the risk function first and then optimizing the perturbed risk function.

The source codes of the \dgen, \pbayes, \pph, \dpc, \pgene were downloaded from \cite{DiffGenImpl}, \cite{PrivBayesImpl}, ~\cite{PPHImpl}, \cite{DiffP-4.5-Impl} and \cite{PrivGeneImpl}, respectively.  The source code of \perm was shared by the authors of \pbayes~\cite{ZCP+14}.

\mypara{Evaluation Methodology.} We consider two baselines --  \emph{Majority} and \emph{NoiseFree}. \emph{Majority} is the misclassification rate by majority voting on the class attribute, which predicts each test case with the majority class label in the train dataset.  \emph{NoiseFree} is the misclassification rate of a decision tree or SVM classifier built on the true data. We expect that a good algorithm to perform better than \emph{Majority}, and gets close to \emph{NoiseFree} as $\epsilon$ increases.

The evaluation is based on two classification models: the CART decision tree classifier and the SVM classifier with radial basis kernel.
Interactive approaches \dpc and \perm build private classifiers directly.  And we use parameters suggested by the corresponding papers.  The non-interactive approaches \pfc, \pph, \dgen, and \pbayes generate private synthetic datasets. To evaluate their performance in terms of decision tree model, we use the rpart~\cite{rpart} library to build decision trees on their generated synthetic datasets. For the evaluation in terms of SVM model, we use the LibSVM package~\cite{CC01a} to build SVM classifiers on the synthetic datasets.  We use the same set of parameters of rpart and LibSVM respectively in evaluating the above non-interactive approaches.


For all the experiments, we vary $\epsilon$ from 0.05 to 1.0.  Similar to the experiment settings of~\cite{FS10, MCFY11, Vinterbo12}, under each privacy budget, we execute 10-fold stratified cross-validation to evaluate the misclassification rate of the above methods.  For each train-test pair, we run the target method 10 times.  We report the average measurements over the 10 runs and the 10-fold crossvalidations.  We set the maximum grid pool size to be 200,000.
The implementation and experiments of \pfc were done in Python 2.7 and all experiments were conducted on an Intel Core i7-3770 3.40GHz PC with 16GB memory.


\subsection{Comparison against Competitors}

\mypara{Comparison on Decision Tree.} Five approaches are involved: \pfc, \dgen, \pbayes,  \pph and \dpc.
Figure~\ref{fig:comparison-dtree} reports their average misclassification rates 
and the corresponding standard deviations.
Clearly, \pfc has the best performance, followed by \dgen, \pph, \dpc.  \pbayes is the poorest in most cases.
The performance of \pfc is also the most robust, as can be seen from the fact that the standard deviation of its misclassification rates is the lowest.








\mypara{Comparison on SVM.} We compare 6 approaches: \pfc, \dgen, \pbayes, \pph, \pgene, and \perm. Figure~\ref{fig:comparison-svm} reports the experimental results. Once again, \pfc has the best performance, followed by \dgen, \pgene, \pph, \perm and \pbayes.

\mypara{Effectiveness of Private Feature Selection.}
In Figure \ref{fig:fs-comparison}, we evaluate our private feature selection method on the US dataset under privacy budget 0.1.  We create a variant of \pfc, called \pfcFSNF, in which the feature selection step of \pfc is noise-free and all the privacy budget is used in remaining steps.  PPH is included in the comparison since it also has a private feature selection step.  We create variants for each of the rest competitors, by adding our proposed feature selection method as preprocessing step which uses 30\% of the total privacy budget.  

From Figure \ref{fig:fs-comparison}, we can see that our \pfc algorithm has close performance to its counterpart (\pfcFSNF).  This justifies that fact that the set of attributes \pfc selects for grid partition is almost as good as those selected by \pfcFSNF and the effectiveness of \pfc mainly comes from the private histogram selection.  We can also see that although \pbayes, \dpc and \perm's performances are improved significantly by doing our private feature selection step, they are still outperformed by \pfc.

\begin{figure*}[!htb]
\begin{tabular}{cc}
	\includegraphics[width = 3.0in]{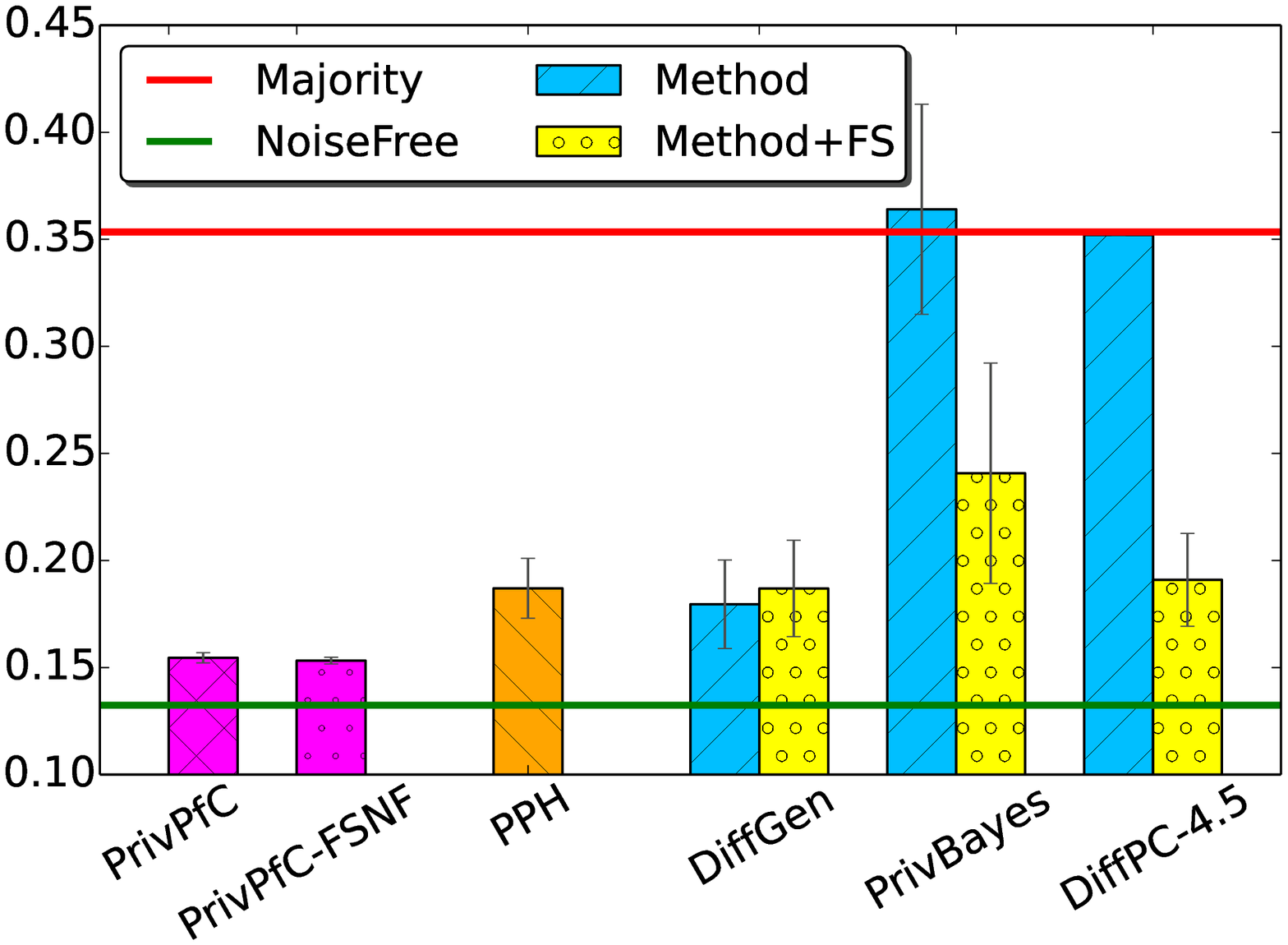}
	& \includegraphics[width = 3.0in]{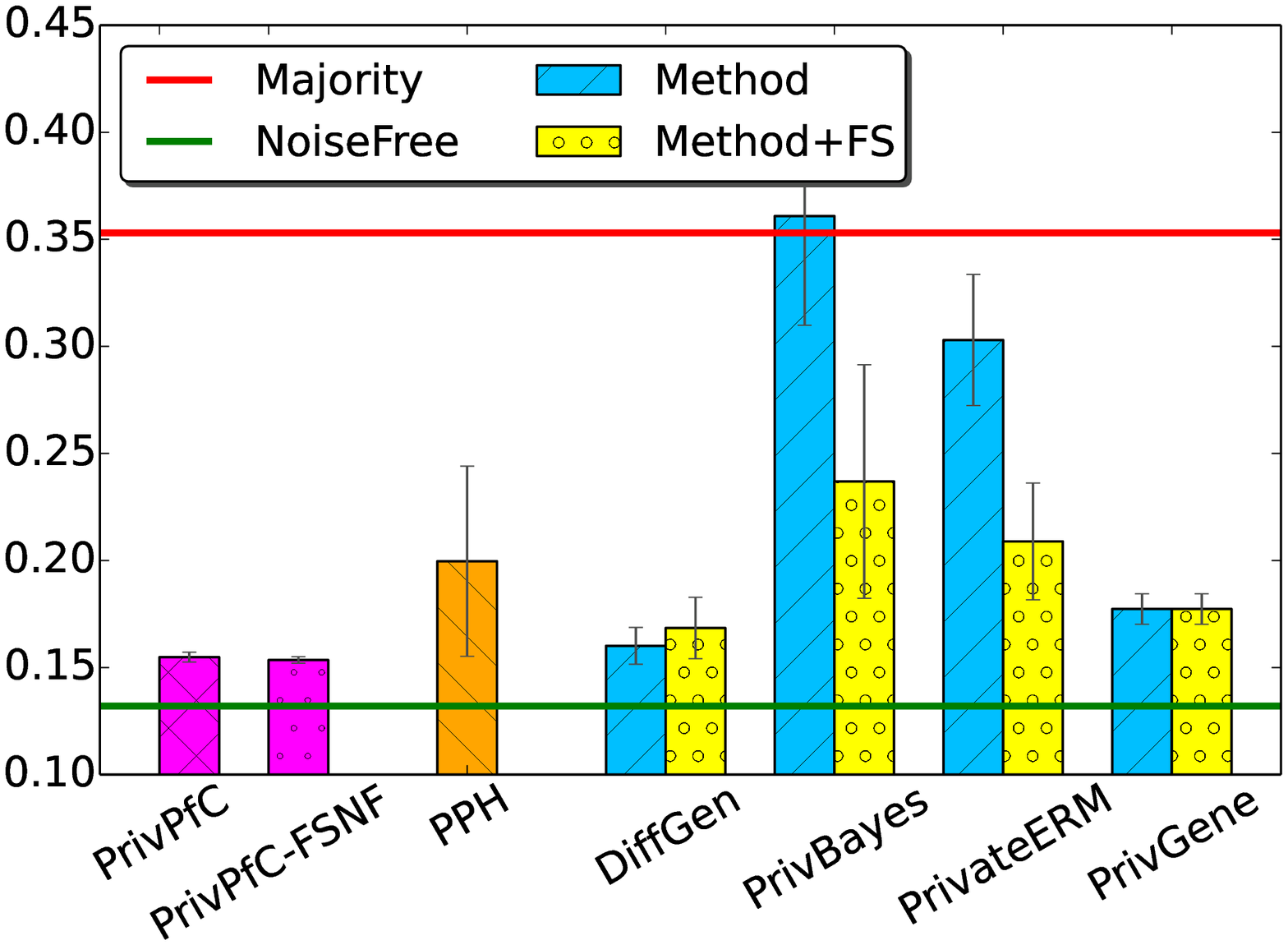}\\
	(a) Decision Tree & (b) SVM\\
\end{tabular}
	\caption{Effectiveness of Private Feature Selection on US dataset with $\epsilon=0.1$. y-axis: misclassification rate.}\label{fig:fs-comparison}
\end{figure*}

\subsection{Analyses of Sources of Errors}

\begin{figure*}[!htb]
\begin{tabular}{cc}

    \multicolumn{2}{c}{\hspace{-0.6cm}\includegraphics[width = \textwidth]{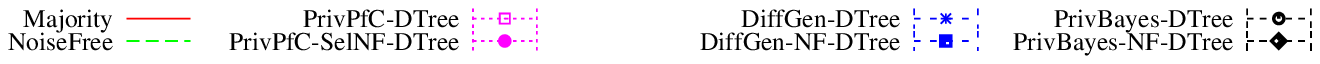}}\\
	\includegraphics[width = 3.0in]{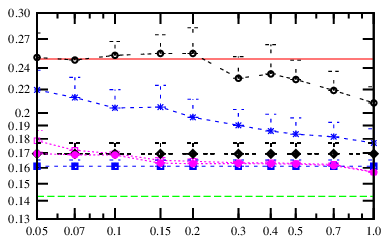}
	& \includegraphics[width = 3.0in]{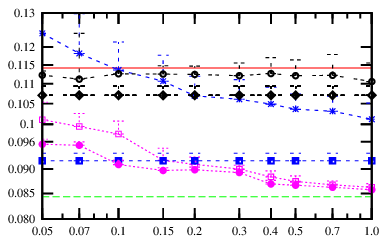}\\
	(a) Adult & (b) Bank\\
	\includegraphics[width = 3.0in]{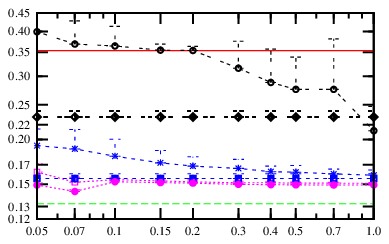}
	& \includegraphics[width = 3.0in]{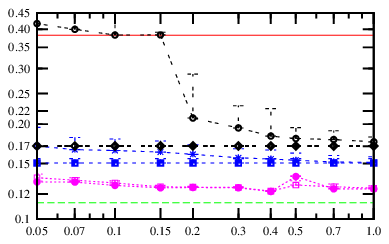} \\
	(c) US & (d) BR\\

\end{tabular}
	\caption{Analyses of \pfc, \dgen and \pbayes by decision tree classification. x-axis: privacy budget $\epsilon$ in log-scale. y-axis: misclassification rate in log-scale.}\label{fig:NF-dtree}
\end{figure*}

\begin{figure*}[!htb]
\begin{tabular}{cc}

    \multicolumn{2}{c}{\hspace{-2mm}\includegraphics[width = \textwidth]{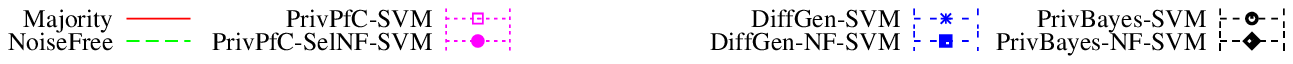}}\\
	\includegraphics[width = 3.0in]{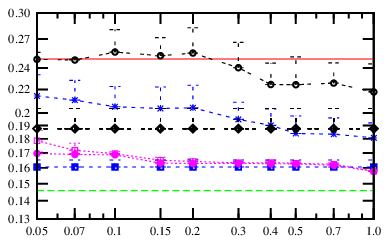}
	& \includegraphics[width = 3.0in]{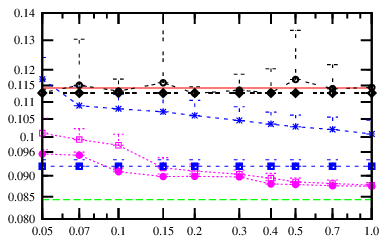}\\
	(a) Adult & (b) Bank\\
	\includegraphics[width = 3.0in]{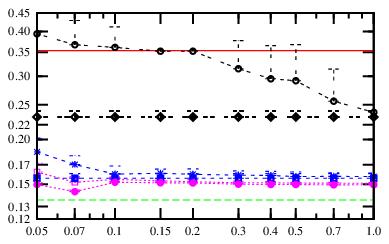}
	& \includegraphics[width = 3.0in]{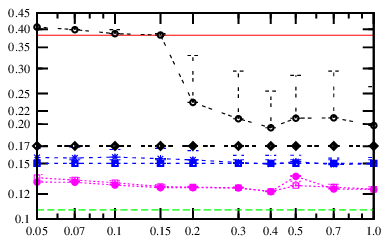} \\
	(c) US & (d) BR\\

\end{tabular}
	\caption{Analyses of \pfc, \dgen and \pbayes by SVM classification. x-axis: privacy budget $\epsilon$ in log-scale. y-axis: misclassification rate in log-scale.}\label{fig:NF-svm}
\end{figure*}

\pfc distributes the privacy budget among three steps, feature selection, grid selection and perturbation, in a 30\%-30\%-40\% way.  When feature selection is not needed, the privacy budget is divided between grid selection and perturbation in a ratio of 3:4.  While these ratios are somewhat arbitrary, we have experimentally evaluated other ratios, allocating between 20\% and 60\% to each step.  We have found that the differences among different budget allocations are minor, so long as the last step receives at least 30\% of the privacy budget.  Even with the worst allocation, which gives 20\% to the last step, \pfc still clearly outperforms competing methods.
We also consider a variant of \pfc, called \pfcSelNF, in which the feature selection step and histogram selection step are noise free and all the privacy budget is used in the histogram perturbation step.  \pfcSelNF is not private; it shows the best one can hope to achieve by optimizing the division among steps.

We have seen that \pfc outperforms the other non-interactive methods such as \dgen and \pbayes.  The key difference in \pfc is that we choose the grid $g$ holistically, instead of arriving at the final grid through a series of decisions.  For example, \dgen iteratively chooses the attributes and ways to partition them, and \pbayes iteratively builds a Bayesian network.
There are two reasons why such an iterative approach does not perform well.  The first is that the decisions made in each iteration may be sub-optimal because of the perturbation necessary for satisfying differential privacy.  The second is that even if the decision made in each iteration is locally optimal, the combination of them is not globally optimal.  To see to what extent the latter factor affects accuracy, we consider noise free variants of them respectively, \dgenNF and \pbayesNF.  In these variants the decisions in each iteration as well as the publishing of counts in the end are performed without any perturbation.  They represent \dgen and \pbayes when the privacy budget $\epsilon$ goes to $\infty$.



Figure~\ref{fig:NF-dtree} and Figure~\ref{fig:NF-svm} report the experimental results of comparing these methods, using Decision Tree and SVM, respectively.  We first observe that while \pfcSelNF indeed outperforms \pfc, the difference is very small, especially for larger $\epsilon$ values in the range.  In fact, on Adult, US, and BR datasets, the difference is barely noticeable when $\epsilon\geq 0.1$.  This suggests that little improvement can be gained to further optimize the division of privacy budget or dataset among determining grid $g$ and publishing noisy histogram.

We also observe that the non-private noise-free version of \pbayes still performs poorly; in fact, it performs significantly worse than the private \pfc.   This suggests that the iterative Bayes network construction approach is not suitable for the purpose of building accurate classifiers.  This is perhaps due in large part to the fact that it is not designed originally to optimize for classification.

The non-private \dgenNF performs similarly to \pfc and \pfcSelNF on the Adult and US datasets.  On the Bank dataset, it is outperformed by \pfc and \pfcSelNF when $\epsilon \geq 0.15$.  On the BR dataset, \dgenNF performs significantly worse than \pfc and \pfcSelNF.  This suggests that the inherent iterative structure of \dgen is suboptimal, even without considering the effect of perturbation.

%% file: conclusion.tex
In this paper, we have introduced \pfc, a novel framework for publishing data for classification under differential privacy.  As a core part of \pfc, we have introduced a novel quality function that enables the selection of a good ``grid'' for publishing noisy histograms.  We have also introduced a new techinque for privately selecting of most relevant features for classification, which enables \pfc to scale to higher-dimension datasets.  We have conducted extensive experiments on four real datasets, and the results show that our approach greatly outperforms several other state-of-the-art methods for private data publishing as well as private classification.

%% file: appendix.tex
\noindent\textbf{Proof of Lemma~\ref{lem:corr-sen}.} 
Without loss of generality, we assume the cell $o_{11}$ is changed by 1.  We use $\mathbf{c}_i$ to denote the number of tuples with $R=r_i$, $i=1,2$. 
\begin{align*}
&\Delta = |\mathrm{corr}(f, D^{\prime}) - \mathrm{corr}(f, D)|\\
		&= \Bigg| \left(o_{11} + 1 - \frac{(o_{11} + o_{12} + 1)(\mathbf{c}_1 + 1)}{\mathbf{c}_1 + \mathbf{c}_2 + 1}\right) - \left(o_{11} - \frac{(o_{11} + o_{12})\mathbf{c}_1}{\mathbf{c}_1 + \mathbf{c}_2}\right)\\
		&+ \sum_{i=2}^{m}\left(o_{i1} - \frac{(o_{i1} + o_{i2})(\mathbf{c}_1 + 1)}{\mathbf{c}_1 + \mathbf{c}_2 + 1}\right) - \left(o_{i1} - \frac{(o_{i1} + o_{i2})\mathbf{c}_1}{\mathbf{c}_1 + \mathbf{c}_2}\right) \Bigg| \\
		& \leq \left|1 - \frac{(o_{11} + o_{12} + 1)(\mathbf{c}_1 + 1)}{\mathbf{c}_1 + \mathbf{c}_2 + 1} + \frac{(o_{11} + o_{12})\mathbf{c}_1}{\mathbf{c}_1 + \mathbf{c}_2}\right|\\
		&+ \sum_{i=2}^{m}\left| - \frac{(o_{i1} + o_{i2})(\mathbf{c}_1 + 1)}{\mathbf{c}_1 + \mathbf{c}_2 + 1} + \frac{(o_{i1} + o_{i2})\mathbf{c}_1}{\mathbf{c}_1 + \mathbf{c}_2}\right|\\
		& = \frac{(\mathbf{c}_1+\mathbf{c}_2)(\mathbf{c}_1+\mathbf{c}_2+1) - (o_{11} + o_{12})\mathbf{c}_2 - (\mathbf{c}_1+1)(\mathbf{c}_1+\mathbf{c}_2)}{(\mathbf{c}_1+\mathbf{c}_2)(\mathbf{c}_1+\mathbf{c}_2+1)}\\
		&+ \frac{\mathbf{c}_2((\mathbf{c}_1-o_{11}) + (\mathbf{c}_2 - o_{12}))}{(\mathbf{c}_1+\mathbf{c}_2)(\mathbf{c}_1+\mathbf{c}_2+1)}\\
		& = \frac{2\mathbf{c}_2((\mathbf{c}_1-o_{11}) + (\mathbf{c}_2 - o_{12}))}{(\mathbf{c}_1+\mathbf{c}_2)(\mathbf{c}_1+\mathbf{c}_2+1)} \\
		&\leq 2. 
\end{align*}
$\Box$